\documentclass[twocolumn]{IEEEtran}
\usepackage{amsmath}
\usepackage{amssymb}
\usepackage{amsfonts}
\usepackage{cite}
\usepackage{epsfig}
\usepackage{enumerate}
\usepackage{algorithm}
\usepackage[all]{xy}
\usepackage{algpseudocode}
\usepackage{xspace}
\usepackage{paralist}
\usepackage{float}

\newtheorem{thm}{Theorem}
\newtheorem{theorem}{Theorem}

\newtheorem{lem}[thm]{Lemma}

\newcommand{\beq}{\begin{equation}}
\newcommand{\eeq}{\end{equation}}
\newcommand{\bea}{\begin{eqnarray}}
\newcommand{\eea}{\end{eqnarray}}
\newcommand{\bean}{\begin{eqnarray*}}
\newcommand{\eean}{\end{eqnarray*}}
\newcommand{\bit}{\begin{itemize}}
\newcommand{\eit}{\end{itemize}}
\newcommand{\ben}{\begin{enumerate}}
\newcommand{\een}{\end{enumerate}}
\newcommand{\blem}{\begin{lem}}
\newcommand{\elem}{\end{lem}}
\newcommand{\bthm}{\begin{thm}}
\newcommand{\ethm}{\end{thm}}

\newcommand{\stale}{stale\xspace}
\newcommand{\Stale}{Stale\xspace}
\newcommand{\stalenode}{s}
\newcommand{\helper}{updated\xspace}
\newcommand{\helpernode}{u}
\newcommand{\supth}{^\textrm{th}}
\newcommand{\msg}{m}
\newcommand{\msgvec}{\boldsymbol{m}}
\newcommand{\msgveca}{\msgvec^{(a)}{}}
\newcommand{\msgvecb}{\msgvec^{(b)}{}}
\newcommand{\msgvecc}{\msgvec^{(c)}{}}
\newcommand{\nodesize}{A}%Total storage per node (including all stripes)
\newcommand{\dparameternobracket}{n-1}%MBR storage per stripe per node; = \nodesize / \numstripes
\newcommand{\dparameter}{(\dparameternobracket)}%MBR storage per stripe per node; = \nodesize / \numstripes
\newcommand{\st}{\stalenode}          % stale node
\newcommand{\up}{\helpernode}          % updated node
\newcommand{\vpsi}[4]{\gamma_{#1,#2} (#3, #4)}   % {\up}{\st}{i}{j}
\newcommand{\sdelta}{\delta} % symbol changed by delta
\newcommand{\etas}[2]{\eta_{#1,#2}} % stripe coefficient {node id, stripe id}
\newcommand{\xis}[2]{\boldsymbol{\xi}_{#1,#2}}%linear combination coefficient for MSR updated nodes {linear combination number (1 or 2),helper node id}
\newcommand{\p}{p} % stripe index
\newcommand{\np}{P} % number of stripes
\newcommand{\numstripes}{\np}
\newcommand{\modified}{modified\xspace}

\newcommand{\fragments}{fragments\xspace}
\newcommand{\genstale}{G_s}
\newcommand{\func}{f}
\newcommand{\genmsr}{\Gamma}
\newcommand{\genmsrs}[1]{\Gamma^{(#1)}}

\emergencystretch=\maxdimen %to force word wrapping

\title{Fundamental Limits on Communication for Oblivious Updates in Storage Networks}
\author{Preetum Nakkiran, Nihar B. Shah, K. V. Rashmi\\Department of EECS, University of California, Berkeley\\{\tt \{preetum,nihar,rashmikv\}@berkeley.edu}\thanks{IEEE Global Communications Conference (GLOBECOM) 2014.}
\vspace{-.11cm}
}

\begin{document}
\maketitle
\thispagestyle{empty}

\begin{abstract}
In distributed storage systems, storage nodes intermittently go offline for numerous reasons. On coming back online, nodes need to update their contents to reflect any modifications to the data in the interim. In this paper, we consider a setting where no information regarding modified data needs to be logged in the system. In such a setting, a `stale' node needs to update its contents by downloading data from already updated nodes, while neither the stale node nor the updated nodes have any knowledge as to which data symbols are modified and what their value is. We investigate the fundamental limits on the amount of communication necessary for such an \textit{oblivious} update process.

We first present a generic lower bound on the amount of communication that is necessary under any storage code with a linear encoding (while allowing non-linear update protocols). This lower bound is derived under a set of extremely weak conditions, giving all updated nodes access to the entire modified data and the stale node access to the entire stale data as side information. We then present codes and update algorithms that are optimal in that they meet this lower bound. Next, we present a lower bound for an important subclass of codes, that of linear Maximum-Distance-Separable (MDS) codes. We then present an MDS code construction and an associated update algorithm that meets this lower bound. These results thus establish the \textit{capacity} of oblivious updates in terms of the communication requirements under these settings.
\end{abstract}

\begin{figure*}[t]
\centering
\includegraphics[width=.9\textwidth]{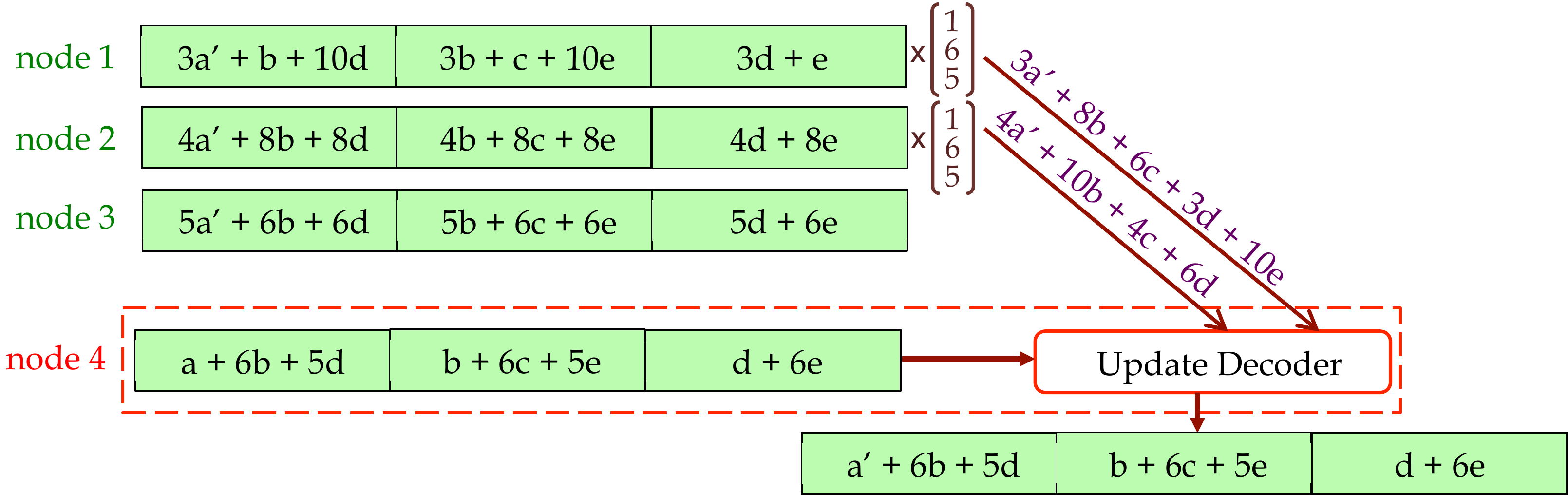}
\caption{A code and an update algorithm that performs optimal oblivious updates. The code operates over the finite field $\mathbb{F}_{11}$. The symbol $a$ was modified to $a'$ during a period when node $4$ was temporarily unavailable/offline. Upon returning, node $4$ updates its stored data despite all nodes being oblivious to the identity and the value of the modified symbol. The update protocol downloads a total of $2 \log_2 11$ bits (i.e., one symbol of $\mathbb{F}_{11}$ each from any two of the \helper nodes), which is the minimum possible.}%\vspace{-.3cm}}
\label{fig:example}
\end{figure*}

\vspace{-.5cm}
\section{Introduction}%TODO: in final version add huang2012erasure to data center citations, and reedSolomon_short to RS, ghemawat2003google_short to GFS and shvachko2010hadoop_short to HDFS
In recent years, there has been a tremendous increase in the amount of digital data stored. This has lead to the popular paradigm of distributed storage wherein the data to be stored is partitioned into \fragments and stored across multiple storage nodes connected through a network. This includes peer-to-peer storage systems~\cite{totalRecall_etal,rowstron2001storage_short,crashplan,spacemonkey}, globally distributed storage systems~\cite{kubiatowicz2000oceanstore_etal,cleversafe_erasure_codes_short}, data-center based storage systems~\cite{ghemawat2003google_short,shvachko2010hadoop_short}, and caching networks~\cite{tang2008benefit_short}. These distributed storage systems store data in a redundant fashion, using either replication or erasure coding, in order to ensure reliability and availability in the face of frequent unavailability events. Under replication, multiple copies of the fragments are stored on different nodes, for example, the Google File System and the Hadoop Distributed File System use 3-replication as the default strategy for introducing redundancy. Under erasure coding, the data fragments are encoded using erasure codes such as Reed-Solomon codes and the encoded fragments are stored on different nodes~\cite{borthakur2008hdfs}.

%For example, nodes join and leave the network at their will in peer-to-peer storage systems; similarly, software issues, maintenance shutdowns, hardware failures, etc., cause unavailability in data center environments~\cite{hotstorage}. Due to these frequent unavailability events, distributed storage systems store data in a redundant fashion, either using replication or erasure coding~\cite{hdfs_raid,}. In replication, multiple copies of the fragments are stored on different nodes, for example, GFS~\cite{} and HDFS~\cite{} use three replication as the default strategy for introducing redundancy. In erasure coding, the data fragments are    encoded using erasure codes such as Reed-Solomon codes~\cite{} and the encoded fragments are stored on different nodes~\cite{HDFS_RAID}.

The storage nodes in the system can go offline for certain intervals of time for various reasons. For instance, there is frequent node churn in peer-to-peer networks as nodes join and leave the network at their will; software issues, maintenance shutdowns, reboots, etc. cause nodes to go offline in distributed storage systems~\cite{ford2010availability_etal,rashmi2013hotstorage_etal}; machines are switched off for certain intervals of time for power savings in some data centers~\cite{lin2013dynamic_short}.

We consider the setting of \textit{mutable} data where the data may be modified during its lifetime (as opposed to immutable data which is read-only). When data gets modified, all stored fragments pertaining to this data (either the replicas or the encoded fragments) need to be updated to reflect this modification. When a node comes back online, its contents need to be updated to reflect any modifications to the data that occurred when the node was offline. We term such a node as a {\it \stale} node.

% for example, due to the usual churn in peer-to-peer networks~\cite{} or due to switching off of nodes for the purpose of saving power~\cite{}. In the meantime, the original data might get \modified, and when the node comes back up, it needs to update its own data to reflect such a \modification.
\newcommand{\centralized}{complete-information\xspace}
One approach towards enabling \stale nodes to update their contents is to centrally track all modifications to the data. Under such a {\it \centralized} approach~\cite{blaum1999lowest_short,xu1999low_short,plank2009raid,tamo2012access_supershort},  the data in a stale node is updated via communication with a central node which provides the precise value of the updated fragments to the stale node. However, this approach has the drawback of requiring the system to centrally keep a log of every modification of the data. This paper, on the other hand, considers an entirely distributed approach in which the system does not store any information regarding the modified data. Here, a \stale node needs to update its contents by communicating and downloading data from other updated storage nodes present in the system. Neither the \stale node nor the updated nodes are aware of what data was modified and what its updated value is. We term such an update process as an \textit{oblivious update}. In this paper, we seek to establish the fundamental limits on the amount of communication required to perform oblivious updates.

In a distributed system, one could constantly store and maintain, in every storage node, a log of all updates. When required to update a \stale node, one could use these logs to identify and transmit the updated data. Oblivious updates, on the other hand, do not necessitate any such additional storage, and also help avoid logistical issues in maintaining any logs. As we will show later in the paper, the amount of communication required to perform an oblivious update is, in fact, not much larger than the amount of communication required for updates in the complete-information setting.

A related line of work is that on maintaining consistency in databases~\cite{vogels2009eventually_short,demers1994bayou_etal} in the presence of modifications to the data. The primary problems here are of ensuring that read requests are served from up-to-date data, and maintaining availability of the data. The problem of set reconciliation~\cite{minsky2003set_short} also has similarities with the problem of oblivious updates. The set reconciliation problem involves two entities, each of whom has some set of values, and the goal is to enable these two entities to learn the difference between their sets with the minimum amount of communication. %In the problem of oblivious updates considered in this paper, a \stale node wishes to update the data stored in it by communicating minimum amount of data with the already updated nodes.

%TODO: in longer version, also include ,jin2009pcode_short,cassuto2009cyclic,xu1999xcode_short
Following the literature on classical \centralized updates~\cite{blaum1999lowest_short,xu1999low_short,plank2009raid,tamo2012access_supershort}, in this paper we study the case when at most a single \textit{symbol} is modified. Here, a `symbol' refers to the smallest granularity of data that can be modified. The case of a single-symbol update is a stepping stone to the more general case of multiple symbol-updates. %, and plan to build on the results of this paper to address the general case in the near future.
 % The centralized settings assume that when a message symbol gets modified, the identities and new values of the updated symbols are immediately sent to the corresponding nodes. Thus in such a setting, the metric of interest is the number of stored symbols that need to be updated when a single message symbol is modified. On the other hand, we are interested in the setting where all nodes are completely oblivious to the modification; no node stores information about what symbol is modified and what the updated value is.
Further, motivated by practical considerations, we restrict our attention to linear codes, i.e., where the encoding process for storage is linear. Although the storage codes are linear, the update protocol is allowed to involve non-linear computations as well, thereby leading to more general bounds.

In this paper we investigate the fundamental limits on the amount of data that needs to be communicated to perform oblivious update of a \stale node when a single message symbol is modified. We show that under any code that has a linear encoding (over a finite field of size $q$), including the special case of `replication', a \stale node needs to download at least $2 \log_2 q$ bits when any one of the message symbols is \modified (Section~\ref{sec:lower_any}). This lower bound is obtained via a genie-based argument under a set of extremely weak conditions allowing infinite connectivity for the \stale node and giving the entire modified data to all the updated nodes and the entire stale data to the stale node as side information. We then present codes and update algorithms that, perhaps surprisingly, meet these lower bounds on communication (Section~\ref{sec:mbr}). Here, oblivious updates are preformed by having a \stale node download only $2 \log_2 q$ bits, while the amount of data stored in the node may be arbitrarily large. These codes are also optimal with respect to the `storage-bandwidth tradeoff' for distributed storage~\cite{dimakis2010network_etal}. We then investigate the class of codes that are `Maximum-Distance-Separable' (MDS). MDS codes are a popular choice for distributed storage since they provide maximum reliability with minimum storage overheads. When the linear code is restricted to be MDS, we establish a lower bound on the amount of communication required for oblivious update (Section~\ref{sec:lower_mds}), and additionally, present an MDS code and an update algorithm that meets this lower bound (Section~\ref{sec:achieve_mds}). These results thus establish the \textit{capacity} of the communication requirements for oblivious updates under linear codes.

The next section formalizes the problem setting and presents an illustrative example.

%The problem setting, an example of the capacity achieving explicit codes

%This paper considers a setting where data is stored across multiple storage nodes distributed across a network. The data to be stored is mutable, i.e., may be modified after it is first stored.~\footnote{This is in contrast to immutable data in several systems such as~\cite{ourHotStorage}.} The storage nodes may be unavailable for intervals of time, for example, due to the usual churn in peer-to-peer networks~\cite{} or due to switching off of nodes for the purpose of saving power~\cite{}. In the meantime, the original data might get \modified, and when the node comes back up, it needs to update its own data to reflect such a \modification. We consider an entirely distributed setup in which no server stores any log regarding any information on the updates; the update thus must be done by communicating with other storage nodes in the distributed network, who will also be oblivious of what was updated when. The goal is to establish the fundamental limits on the amount of communication required to perform this update procedure.

\section{Problem Description}
\subsection{Problem Setting}\label{sec:model}
Consider $B$ symbols of data, termed the \textit{message}, that are to be stored across $n$ storage nodes. Each symbol of data is assumed to belong to some finite field $\mathbb{F}_q$ of size $q$. Each node has a capacity of storing $\nodesize \geq 2$ symbols over $\mathbb{F}_q$. The data is stored across the nodes using a code that is linear over $\mathbb{F}_q$. %The code must be such that for some predefined parameter $k$, the data should be recoverable from \textit{any} $k$ of the $n$ nodes. We consider $n>k>1$ and do not consider the degenerate cases of $n=k$ and $k=1$.
Now suppose some storage node, say node $\stalenode$, was busy or offline for some period of time. In this period, suppose one of the $B$ message symbols was updated. The remaining nodes now store (encodings of) the updated data. However, node $\stalenode$ still contains stale data, and we will call this node as the the \textit{\stale node}. Now, when node $\stalenode$ comes back up, its contents must be updated to reflect the updated message. To this end, the \stale node connects to one or more of the other nodes, and downloads some functions of the data stored in them. The goal is to minimize this amount of download.

In the setup we consider, none of the nodes are required to store any information about the identity or the value of the symbol that was updated. The update of the \stale node's data is thus \textit{oblivious} of the update in the message. We do assume, however, that the \stale node knows that at most a single symbol was updated. We also assume that the code is linear, i.e., each nodes stores $\nodesize$ linear combinations of the $B$ message symbols. Note that we only assume that the underlying encoding of the stored data is linear, and the data passed during an update operation may comprise arbitrary (linear or non-linear) functions of the data stored in the \helper nodes.

%TODO in longer version, give RS reference
The second half of the paper considers a very popular subclass of codes known as \textit{Maximum-Distance-Separable (MDS)} codes. MDS codes satisfy the two following properties:
(a) The entire message of $B$ symbols can be recovered from the data stored in \textit{any} $k$ of the $n$ nodes, for some pre-defined parameter $k$.  This ensures that the storage system can tolerate the failure of any arbitrary $(n-k)$ of the $n$ nodes, and furthermore, ensures high availability of the data since it can be recovered from any $k$ nodes. (b) The storage requirement at each node is $\nodesize = \frac{B}{k}$, which is the minimum possible when satisfying the first property. Again, the goal is to minimize the amount of download required to perform an oblivious update.

%As an aside, observe that a protocol for oblivious-updates of at most $t$ updates can alternatively be used to correct corruption of $t$ or fewer symbols in the data stored within any storage nodes, by treating the corrupted data as `stale data' and the correct data in the other nodes as the `updated data'.

\textit{Notational conventions:} For vectors $\mathbf{v}_1$ and $\mathbf{v}_2$ of equal lengths, $d_H(\mathbf{v}_1,\mathbf{v}_2)$ will denote the Hamming distance between them. For any positive integer $r$, $[r]$ will denote the set $\{1,\ldots,r\}$. Vectors will be column vectors by default.%, and will be denoted in boldface.% The transpose of a vector or a matrix will be indicated via a superscript ${}^T$.
%\vspace{-.1cm}

\subsection{Example}\label{sec:example}
We illustrate the problem setting with an example of a storage code and an update algorithm that are optimal. The code, shown in Fig.~\ref{fig:example}, operates in the finite field $\mathbb{F}_{11}$ of size $11$. The message comprises $B=5$ symbols $\{a,b,c,d,e\}$, each drawn from $\mathbb{F}_{11}$. The message is encoded and stored across $n=4$ storage nodes as shown in the figure. One can verify that the entire message is recoverable from \textit{any} two of the four nodes, thus making the storage system tolerant to the failure of any two of the four nodes.

Now suppose node $4$ was unavailable for some period of time, during which message symbol $a$ was updated to $a'$. The three other nodes store the updated data, and (`stale') node~$4$ must update its own data by downloading data from the three other nodes. The nodes do not keep any record of what is updated and by how much, i.e., do not know that symbol $a$ was updated and that its new value is $a'$. The update protocol is thus required to be oblivious of the update.

The lower bounds derived subsequently in Section~\ref{sec:lower_any} dictate the necessity of downloading at least $2 \log_2 11$ bits for the update. The following update protocol meets this lower bound, with the \stale node downloading one symbol of $\mathbb{F}_{11}$ each from two arbitrary \helper nodes.
The \stale node contacts \textit{any} two other nodes, say nodes $1$ and $2$, and asks for the inner product of their respective data with $[1~~6~~5]$. The two nodes return the values of $(3a' + 8b + 6c + 3d + 10e)$ and $(4a' + 10b + 4c + 6d)$ respectively. Of course, the \stale node does not know that this received data is computed with $a'$ and not $a$. Next, the \stale node computes an inner product of its own data with $[3~~1~~10]$ to get the value of $(3a + 8b + 6c + 3d + 10e)$, and an inner product of its own data with $[4~~8~~8]$ to get the value of $(4a + 10b + 4c + 6d)$. Subtracting these from the data received from the two other nodes, the \stale node obtains the values of $\{3(a'-a),4(a'-a)\}$. If both these values are zero, then no symbol was updated, and the algorithm terminates. If not, then the algorithm continues in the following manner. Since the identity of the updated symbol is not known, from the perspective of the \stale node, these two values could correspond to either $\{3(a'-a),4(a'-a)\}$ or $\{8(b'-b),10(b'-b)\}$ or $\{6(c'-c),4(c'-c)\}$ or $\{3(d'-d),6(d'-d)\}$ or $\{10(e'-e),0(e'-e)\}$. The \stale node now takes the ratio of the two values; this ratio $3:4$ uniquely identifies that symbol $a$ was updated. Multiplying the first value $3(a'-a)$ by $3^{-1}$ gives the value of the update $(a'-a)$. This amount is added to the first symbol of the \stale node, and the result $(a'+6b+5d)$ is stored as the updated first symbol of the \stale node.

This storage code and update protocol are generalized to arbitrary system parameters in Section~\ref{sec:MBR_encoding}.

\section{Lower Bounds for Arbitrary Linear Codes}\label{sec:lower_any}
This section derives lower bounds on the amount of download for oblivious update under any arbitrary code with linear encoding. Note that although we consider the encoding to be linear, we allow the update operation to be executed via any arbitrary (linear or non-linear) functions.

\begin{theorem}\label{thm:lower_any}
Consider a scenario where the \stale node is allowed to connect to any arbitrary number of updated nodes. Furthermore, suppose a genie provides all \helper nodes with all the updated message symbols and the \stale node with all the \stale message symbols as side information (the nodes still do not know the identity or value of the symbol that was updated). In order to update the data stored in the \stale node, it must download a total of at least $2\log_2 q$ bits.
\end{theorem}
\begin{IEEEproof}
Let $\msgvec$ denote the vector of the $B$ \stale symbols, and let $\msgvec'$ denote the vector of the $B$ updated symbols, with $d_H(\msgvec,\msgvec')\leq 1$. Let $\genstale$ denote the $(\nodesize \times B)$ generator matrix of the \stale node, i.e., the \stale node stores $\genstale \msgvec$, and wants $\genstale \msgvec'$. Assume without loss of generality that the $\nodesize$ rows of $\genstale$ are linearly independent.  Note that our genie has also provided the entire \stale message $\msgvec$ to the \stale node.

Since the genie has provided each \helper node with all the updated message symbols $\msgvec'$, one can assume without loss of generality that the \stale node connects to only one \helper node. On being contacted by the \stale node, the \helper node must return some function of the data: let $\func$ denote this function, i.e., the \helper node returns $\func(\msgvec')$ to the \stale node. We will now show that the cardinality of the range of $\func$ cannot be less than $q^2$, thus necessitating a download of at least $2\log_2 q$ bits. We employ a contradiction-based argument, for which we assume that the cardinality of the range of $\func$ is strictly smaller than $q^2$.

The linear independence of the $\nodesize$ rows of $\genstale$ implies existence of $\nodesize$ coordinates $\ell_1,\ldots,\ell_\nodesize \in [B]$ of $\msgvec$ such that for every fixed value of $\msgvec \backslash \{\msg_{\ell_1},\ldots,\msg_{\ell_\nodesize}\}$, the map $(\msg_{\ell_1},\ldots,\msg_{\ell_\nodesize}) \rightarrow \genstale \msgvec$ is a bijection. Without loss of generality, let $\ell_i=i~\forall i \in [\nodesize]$.
Consider the set of $q^2$ messages of the form $[\msg_1', \msg_2', 0\ldots,0]$.
Since the range of $\func$ contains strictly fewer than $q^2$ values,
there must exist some two distinct messages, say $\msgveca$ and $\msgvecb$, in the aforementioned set of $q^2$ messages, for which
$\func(\msgveca)=\func(\msgvecb)$.

Now we know that $\msgveca \neq \msgvecb$, $\func(\msgveca)=\func(\msgvecb)$, $\genstale \msgveca \neq \genstale \msgvecb$, and $d_H(\msgveca ,\msgvecb)\leq 2$.
The last property implies existence of some $\msgvecc \in (\mathbb{F}_q)^B$ such
that $d_H(\msgvecc ,\msgveca)\leq 1$ and $d_H(\msgvecc,\msgvecb)\leq 1$.
Finally, suppose $\msgvecc$ is the \stale message.
Now, $\msgveca$ and $\msgvecb$ are two possible candidates for the updated message.
The \stale node has access to the same data in both cases: $\genstale \msgvecc$ as its own \stale data, and $\func(\msgveca)=\func(\msgvecb)$ downloaded from the \helper node.
This prevents the \stale node from distinguishing between $\msgveca$ and $\msgvecb$ as the
updated message.
However, the updated data at the stale node must be different
(since $\genstale \msgveca \neq \genstale \msgvecb$),
making it necessary to distinguish between the two cases. This causes a contradiction.
\end{IEEEproof}

\section{Codes Achieving Lower Bounds}\label{sec:mbr}
The lower bound derived in Theorem~\ref{thm:lower_any} is in the presence of a very helpful genie. This section presents codes and update algorithms that meet this bound in the absence of this genie. %This establishes the \textit{capacity} of oblivious updates under linear codes.
These upper bounds are obtained by proving the existence of codes meeting these bounds, and towards this, we employ the product-matrix framework of~\cite{rashmi2011productmatrix_supershort}. Interestingly, the proposed codes are also optimal with respect to the \textit{storage-bandwidth tradeoff} derived in~\cite{dimakis2010network_etal}. The update algorithms presented here require the stale node to connect to \textit{any two} updated nodes.

\subsection{Encoding}\label{sec:MBR_encoding}
The code is associated to an additional parameter $k \in [n-1]$, and has the
property that the entire message can be recovered from any $k$ of the nodes.
Assume that $B$ is divisible by $\left(k(n-1)-\frac{k(k-1)}{2}\right)$.\footnote{If not, then append an appropriate number of zeros to the message. Since the amount of data $B$ is typically much larger than $n$ and $k$, this operation is relatively inexpensive.}
\;Let $\numstripes := \frac{B}{\left(k(n-1)-\frac{k(k-1)}{2}\right)}$.

Under the proposed code, each node is required to store $\nodesize =  \dparameter\numstripes$ symbols over $\mathbb{F}_q$. The value of $q$ will be specified later.

Construct $\numstripes$ \textit{symmetric} matrices $\left\{M_\p\right\}_{\p \in [\np]}$, each of size $(\dparameter \times \dparameter)$, in the following manner. In each matrix $\{M_\p\}_{\p \in [\np]}$, set the bottom-right $((\dparameternobracket-k)\times(\dparameternobracket-k))$ submatrix to zero. Each of these (symmetric) matrices now have $\left(k(n-1)-\frac{k(k-1)}{2}\right)$ free elements remaining. Partition the $B$ message symbols into $\numstripes$ sets of $\left(k(n-1)-\frac{k(k-1)}{2}\right)$ symbols each.  For each $\p \in [\np]$, populate the remaining free elements of matrix $M_\p$ with the message symbols of the $\p\supth$ set.

%Construct $\numstripes$ matrices $\left\{\Psi_\p\right\}_{\p \in [\np]}$, each of size $(n \times \dparameter)$ in the following manner. These matrices are constructed such that its entries satisfy certain conditions that will be outlined subsequently. For $i \in [n]$, let $\boldsymbol{\psi}_i^T$ denote the $i\supth$ row of the matrix $\Psi$.

Construct vectors $\left\{\boldsymbol{\psi}_\ell \right\}_{\ell \in [n]}$, each of length $\dparameter$, and scalars $\left\{\eta_{\ell,\p} \right\}_{\ell \in [n], \p \in [\np]}$ that satisfy:\\
(a) every submatrix of $\left[\boldsymbol{\psi}_1~\cdots~\boldsymbol{\psi}_n\right]$ is of full rank\\
(b) for every $(\up_1,\up_2,s) \in [n]^3$ such that $\up_1\neq \up_2 \neq \st$, and every $(\p,i,j)\in [\np] \times [\dparameter]^2$ and $(\p',i',j')\in [\np] \times [\dparameter]^2$ such that  $(\p,i,j) \neq (\p',i',j')$,
\begin{multline}
{\etas{\up_1}{\p}\vpsi{\up_1}{\st}{i}{j}} {\etas{\up_2}{\p'}\vpsi{\up_2}{\st}{i'}{j'}} \\ \neq {\etas{\up_1}{\p'}\vpsi{\up_1}{\st}{i'}{j'}}{\etas{\up_2}{\p}\vpsi{\up_2}{\st}{i}{j}}~,\nonumber
\end{multline}
where
\beq
\vpsi{\up}{\st}{i}{j} :=
\begin{cases}
\psi_{\up,i} \psi_{\st,j}+\psi_{\up,j} \psi_{\st,i} &\quad \textrm{if}\quad i \neq j\\
\psi_{\up,i} \psi_{\st,i} &\quad \textrm{otherwise}~.
\end{cases}
\label{eq:gamma}
\eeq
Each of these requirements is equivalent to showing that a product of polynomials is non-zero. One can see that each of these polynomials individually is a non-zero polynomial. The Schwartz-Zippel lemma ensures that there exist values of $\left\{\boldsymbol{\psi}_\ell \right\}_{\ell \in [n]}$ and $\left\{\eta_{\ell,\p} \right\}_{\ell \in [n], \p \in [\np]}$ satisfying all the desired conditions when the size $q$ of the underlying finite field $\mathbb{F}_q$ is large enough. Finally, for every $\ell \in [n]$, node $\ell$ stores the data
\[
\left\{ \boldsymbol{\psi}_\ell^T M_\p \right\}_{\p \in [\np]}~.
\]
Condition (a) will help in recovery of the entire message from any $k$ of the nodes, and condition (b) will help in performing the oblivious updates.
%TODO Put references to schwarz zippel lemma in full version; in globecom version if there is space

%The vectors $\left\{\boldsymbol{\eta}_\ell \right\}_{\ell \in [n]}$ will be employed in the update algorithm.

%to the downloads required to repair failed nodes and the storage requirements.
%In~\cite{dimakis2010network_etal}, Dimakis et al. introduced the ``regenerating codes'' model for distributed storage systems, and derived a certain tradeoff between the storage requirements and the network-bandwidth utilization in such systems. The two most important regimes in this tradeoff are the two extremes: one extreme of the tradeoff minimizes the bandwidth utilization and is called the `minimum bandwidth regenerating (MBR)' regime, while the other minimizes the storage requirement (resulting in MDS codes) and is called the `minimum storage regenerating (MSR)' regime. In this section, we show that codes operating in the MBR regime can achieve the genie-based lower bounds of Theorem~\ref{thm:lower_any}.

%The construction presented in this section employs the product-matrix framework proposed in~\cite{rashmi2011productmatrix_supershort}. The product-matrix framework was employed to construct explicit codes for both, the MBR and the MSR regimes in~\cite{rashmi2011productmatrix_supershort}. Our codes will employ the structure of the product-matrix MBR codes.

\subsection{Oblivious Update Algorithm and Performance}
\begin{theorem}\label{thm:achievable_MBR}
In the code constructed in Section~\ref{sec:MBR_encoding}, any \stale node can perform an oblivious update by downloading one symbol each from {any} two updated nodes when at most one symbol has changed.
\end{theorem}
\begin{IEEEproof}
Let $\{M_\p\}_{\p \in [\np]}$ be the matrices comprising the \stale message, as constructed in Section~\ref{sec:MBR_encoding}. The construction is such that no two matrices in $\{M_\p\}_{\p \in [\np]}$ have any element in common. As a result, the update of a single element causes a change in only one of these matrices. Let $\{M_\p'\}_{\p \in [\np]}$ be the matrices comprising the updated message.
Algorithm~\ref{algo:comm_MBR} updates the data of a stale node by connecting to any two updated nodes and downloading only one symbol from each. Recall that the notation $\vpsi{\cdot}{\cdot}{\cdot}{\cdot}$ used Step 2 onwards is defined in~\eqref{eq:gamma}. Step 6  employs condition (b) of the encoding which guarantees
${\etas{\up_1}{\p}\vpsi{\up_1}{\st}{i}{j}} \neq 0$.
\end{IEEEproof}

\begin{algorithm}[h]
\begin{itemize}
\item[]
\vspace{.2cm}
\hspace{-.85cm} \Stale node $\st$ contacts \textit{any} two updated nodes $\up_1$ and $\up_2$
\vspace{.2cm}
%\settowidth{\itemindent}{\bf Updated Nodes:}
\item[\bf Updated Nodes:]
Node $\up_i$ ($i \in \{1,2\})$, which stores updated data $\{\boldsymbol\psi_{\up_i}^T M'_\p\}_{\p \in [\np]}$, returns the single symbol
$\sum_{\p=1}^\np \etas{\up_i}{\p} \boldsymbol\psi_{\up_i}^T M'_\p \boldsymbol\psi_\st$
\vspace{.3cm}
%\settowidth{\itemindent}{\bf Stale Node:}
\item[\bf Stale Node:]
Stale node $\st$, which stores stale data $\left\{\boldsymbol\psi_\st^T M_\p\right\}_{\p \in [\np]}$, receives the two symbols
\begin{alignat*}{2}
r_1' := \sum_{\p=1}^\np \etas{\up_1}{\p} \boldsymbol\psi_{\up_1}^T M'_\p \boldsymbol\psi_\st~,
~~~ &&
r_2' := \sum_{\p=1}^\np \etas{\up_2}{\p} \boldsymbol\psi_{\up_2}^T M'_\p \boldsymbol\psi_\st.
\end{alignat*}
It performs the following operations.
\begin{enumerate}
\item From its stale data, compute:
\begin{alignat*}{2}
\!\!\!\!\!\! r_1 := \sum_{\p=1}^\np \etas{\up_1}{\p} \boldsymbol\psi_{\up_1}^T M_\p \boldsymbol\psi_\st~,
~~~ &&
r_2 := \sum_{\p=1}^\np \etas{\up_2}{\p} \boldsymbol\psi_{\up_2}^T M_\p \boldsymbol\psi_\st
\end{alignat*}

\item Subtract these from the received symbols to get
$d_1 := r_1' - r_1$ and
$d_2 := r_2' - r_2$
If the changed symbol is at location $(i, j)$ of matrix $M_\p$, and its value
has been changed by $\sdelta$, then
$d_1 = \etas{\up_1}{\p}\vpsi{\up_1}{\st}{i}{j} \sdelta$ and
$d_2 = \etas{\up_2}{\p}\vpsi{\up_2}{\st}{i}{j} \sdelta$

\item If $d_1 = d_2 = 0$ then the stale node already has the updated data; exit

\item Compute the ratio
$d_1:d_2 =
{\etas{\up_1}{\p}\vpsi{\up_1}{\st}{i}{j}}:{\etas{\up_2}{\p}\vpsi{\up_2}{\st}{i}{j}}$

\item Condition (b) ensures that this ratio is different for different $(\p, i, j)$, so use the ratio to
identify changed location $(i_0, j_0)$ and $\p_0$.
\item Compute $\sdelta =  (\etas{\up_1}{\p_0} \vpsi{\up_1}{\st}{i_0}{j_0})^{-1} d_1$
\item Construct an $(\dparameter \times \dparameter)$ matrix $\Delta$
with value $\sdelta$ at locations $(i_0, j_0)$ and $(j_0, i_0)$ and zeros elsewhere;
in the stale node, update data $\boldsymbol\psi_\st^T M_{\p_0} $ to $\boldsymbol\psi_\st^T M'_{\p_0}$ as
$\boldsymbol\psi_\st^T M'_{\p_0} = \boldsymbol\psi_\st^T M_{\p_0} +
\boldsymbol\psi_\st^T \Delta$
\end{enumerate}
\end{itemize}
\caption{Optimal Oblivious Update}\label{algo:comm_MBR}
\end{algorithm}

\begin{theorem}
In the code constructed in Section~\ref{sec:MBR_encoding}, the message can be recovered from the data stored in \textit{any} $k$ nodes. Furthermore, the code is optimal with respect to the storage-bandwidth tradeoff of~\cite{dimakis2010network_etal}.
\end{theorem}
\begin{IEEEproof}
The code falls under the `product-matrix MBR' framework of~\cite[Section IV]{rashmi2011productmatrix_supershort} from which it derives these properties.
\end{IEEEproof}

\section{Lower Bounds for Linear MDS Codes}\label{sec:lower_mds}
In this section, we consider the class of codes that are `Maximum-Distance-Separable (MDS)' (recall definition from the last paragraph of Section~\ref{sec:model}). We provide lower bounds on the amount of download for arbitrary MDS codes with linear encoding. Although we consider the encoding to be linear, we allow the update operation to be executed via any arbitrary (linear or non-linear) functions.

\begin{theorem}\label{thm:lower_MDS}
Under any MDS code with linear encoding, a \stale node must contact at least $k$ updated nodes. Upon contacting $k$ nodes, the \stale node must download at least $2\log_2 q$ bits from each them.
\end{theorem}
\begin{IEEEproof}
\newcommand{\setS}{\mathcal{S}}
We will first show that an oblivious update cannot be performed by contacting just $(k-1)$ nodes. The proof is by contradiction for which we will assume existence of some $(k-1)$ nodes from which some stale node can be updated. Suppose the entire data stored in these $(k-1)$ nodes is made available to the stale node. Since the code is MDS, there exists exactly one message whose encoding equals the data currently stored in these $(k-1)$ \helper nodes and the \stale node. The stale node will thus be unable to distinguish between the two cases: (a) this message as the \stale message and no update, and (b) the actual \stale and updated messages. The updated data at the \stale node must be different in the two cases, thus necessitating it to distinguish the two cases. This yields a contradiction.

Now assume the \stale node connects to some $k$ nodes. We now show that it must download at least $2\log_2 q$ bits from each of these $k$ nodes. It suffices to show that the last of these $k$ \helper nodes must pass $2\log_2 q$ bits, since any of these $k$ nodes may be defined as the last node. To this end, consider a genie who provides the entire data stored in the first $(k-1)$ \helper nodes to the \stale node, and furthermore, provides the entire updated message to the last \helper node.

Let $\msgvec\in (\mathbb{F}_q)^B$ be the \stale message, and $\msgvec' \in (\mathbb{F}_q)^B$ be the modified message (with $d_H(\msgvec,\msgvec') \leq 1)$. Let $\genstale$ denote the $(\nodesize \times B)$ generator matrix of the \stale node, i.e., the \stale node stores $\genstale \boldsymbol{m}$ under message $\boldsymbol{m}$. Assume without loss of generality that the $\nodesize$ rows of $\genstale$ are linearly independent. Upon being contacted by the \stale node, the last \helper node (to whom the genie has provided all the updated data) must send some function of the data: let $\func$ denote this function, i.e., the \helper node returns $\func(\msgvec')$ to the \stale node. We will now show that the range of $\func$ must contain at least $q^2$ elements, thus necessitating a download of at least $2\log_2 q$ bits.

The linear independence of the $\nodesize$ rows of $\genstale$ implies existence of $\nodesize$ coordinates $\ell_1,\ldots,\ell_\nodesize \in [B]$ of $\msgvec$ such that for every fixed value of $\msgvec \backslash \{\msg_{\ell_1},\ldots,\msg_{\ell_\nodesize}\}$, the map $(\msg_{\ell_1},\ldots,\msg_{\ell_\nodesize}) \rightarrow \genstale \msgvec$ is a bijection. Without loss of generality, let $\ell_i=i~\forall i \in [\nodesize]$.

%Partition the set of all possible $q^\nodesize$ values of $[\msg_1',\ldots,\msg_\nodesize',0\ldots,0]$ into sets that map on to identical values in the range of $\func$. Since the range of $\func$ contains strictly fewer than $q^2$ values, at least one of these sets must contain strictly more than $q^{\nodesize - 2}$ values. Let us call this set as $\mathcal{R}''$.

% Since the first $(\nodesize - 2)$ coordinates of $[\msg_1',\ldots,\msg_\nodesize',0,\ldots,0]$ can only take on $q^{\nodesize - 2}$ different values, it follows that there must necessarily exist at least two elements in this set with identical values in the first $(\nodesize - 2)$ coordinates, and hence a Hamming distance no more than $2$. Let us denote these two elements as $\msgveca$ and $\msgvecb$.

Let $\setS'$ denote the set of all $q^\nodesize$ messages of the form
$[\msg_1',\ldots,\msg_\nodesize',0\ldots,0]$. Construct a second set $\setS''$
from $\setS'$ in the following manner. For each $\msgvec' \in \setS'$, find the
unique vector $\msgvec'' \in (\mathbb{F}_q)^B$ such that $\genstale \msgvec'' =
\genstale \msgvec'$ and the encoding of $\msgvec''$ in the first $(k-1)$
\helper nodes is zero. Since the code is MDS, for each $\msgvec'$, there exists exactly one such $\msgvec''$. Set $\setS''$ as the collection of these vectors $\msgvec''$.

Partition the set $\setS''$, of size $q^\nodesize$, into sets that map onto identical values in the
range of $\func$. Since the range of $\func$ has a cardinality strictly smaller than $q^2$,
at least one of these sets must have a cardinality strictly greater than $q^{\nodesize - 2}$.
Let us call this set $\mathcal{R}''$.

Now consider the original elements $\mathcal{R}' \subseteq \setS'$ which were transformed into $\mathcal{R}''$. In this set $\mathcal{R}'$, of size greater than $q^{\nodesize - 2}$, there must exist some two messages $\msgveca'$ and $\msgvecb'$ which match on the first $(\nodesize-2)$ coordinates.
It follows that there exists $\msgvecc' \in (\mathbb{F}_q)^B$ such that $d_H(\msgveca',\msgvecc') \leq 1$ and $d_H(\msgvecb',\msgvecc') \leq 1$. Next, let $\msgveca''$ and $\msgvecb''$ respectively be the (distinct) constituents of $\mathcal{R}''$ that are derived from $\msgveca'$ and $\msgvecb'$ respectively.

Finally, consider the following scenario. Suppose the original message was
$(\msgveca''-\msgveca'+\msgvecc')$, and this was updated to $\msgveca''$. This
constitutes the update of at most one symbol since
$d_H(\msgveca''-\msgveca'+\msgvecc',\msgveca'')=d_H(\msgveca',\msgvecc') \leq
1$. We claim that this scenario is indistinguishable from the scenario of the
original message being $(\msgvecb''-\msgvecb'+\msgvecc')$ and the updated message being
$\msgvecb''$. To this end, first observe that the latter situation also
constitutes the update of at most one symbol since
$d_H(\msgvecb''-\msgvecb'+\msgvecc',\msgvecb'')=d_H(\msgvecb',\msgvecc') \leq
1$. Furthermore, since $\genstale\msgveca'=\genstale\msgveca''$ and $\genstale
\msgvecb'=\genstale\msgvecb''$, it must be that the encoding
$\genstale(\msgveca''-\msgveca'+\msgvecc')$ of
$(\msgvecb''-\msgvecb'+\msgvecc')$ at the \stale node is identical to the
encoding $\genstale(\msgvecb''-\msgvecb'+\msgvecc')$ of
$(\msgvecb''-\msgvecb'+\msgvecc')$ in the \stale node. The data stored in the
stale node thus provides no information pertaining to distinguishing these two
scenarios. As discussed above, the encoding of $\msgveca''$ and $\msgvecb''$
both result in zeros at the first $(k-1)$ helper nodes. Furthermore,
$\msgveca'', \msgvecb'' \in \mathcal{R}'' \implies \func(\msgveca'')=\func(\msgvecb'')$ which makes the data downloaded from the last \helper node identical in the two cases. An accurate update is thus impossible in this situation, thus proving our claim.
\end{IEEEproof}

\section{MDS Codes Achieving Lower Bounds}\label{sec:achieve_mds}
In this section, we present upper bounds on the amount of download required for oblivious updates under MDS codes, that meet the lower bounds established in Theorem~\ref{thm:lower_MDS}.  %This establishes the \textit{capacity} of oblivious updates under linear MDS codes.

\subsection{Encoding}\label{sec:msr_encoding}
Each node has a storage capacity of $\nodesize := \frac{B}{k}$ symbols. Let $\boldsymbol{m}$ be a $B$-length vector consisting of the $B$ message symbols. Let $\genmsr$ be an arbitrary $(n\nodesize \times B)$ matrix with the property that every submatrix of $\genmsr$ is of full rank. For instance, one can choose $\genmsr$ as a Cauchy matrix.
Construct $n$ matrices $\{\genmsr_\ell \}_{\ell \in [n]}$, each of size $(\nodesize \times B)$,
by partitioning $\genmsr$ into $n$ blocks of $\nodesize$ rows each.
%Notice the following properties:\\
%(a) the $k\nodesize$ rows of any $k$ of the matrices in the set $\{\genmsr_\ell \}_{\ell \in [n]}$ are all linearly independent\\
%(b) every submatrix of $\genmsr_\ell$ is of full rank $\forall \ell \in [n]$.
Finally, for every $\ell \in [n]$, node $\ell$ stores the data
\[
\genmsr_\ell \boldsymbol{m}~.
\]
\subsection{Oblivious Update Algorithm and Performance}
\begin{theorem}\label{thm:achievable_MDS}
In the code constructed in Section~\ref{sec:msr_encoding}, any \stale node can perform an oblivious update by downloading $2\log_2 q$ bits each from {any} $k$ \helper nodes when at most one symbol has changed.
\end{theorem}
\begin{IEEEproof}
Let $\msgvec \in (\mathbb{F}_q)^B$ be the \stale message and let $\msgvec' \in (\mathbb{F}_q)^B$ be
the updated message, with $d_H(\msgvec,\msgvec')\leq 1$.
For every $\ell \in [n]$, let $\genmsrs1_\ell$ and $\genmsrs2_\ell$ be the first and second rows of
$\genmsr_\ell$, respectively. Further, for any $m \in \{1,2\}$ and any $j \in [B]$, let $(\genmsr^{(m)}_\ell)_{j}$ denote the $j^{\rm th}$ element of $\genmsr^{(m)}_\ell$.

Algorithm~\ref{algo:comm_MSR} updates the data of a stale node by connecting to any $k$ updated nodes and downloading exactly two symbols from each.
Steps 5 and 6 of Algorithm~\ref{algo:comm_MSR} employ the fact that every submatrix of $\genmsr$ is of full rank.
\end{IEEEproof}

\begin{algorithm}[h]
\begin{itemize}
\item[]
\vspace{.2cm}
\hspace{-.85cm} \Stale node $\st$ contacts \textit{any} $k$ updated nodes
$u_1,\hdots,u_k$.
\vspace{.2cm}
\item[\bf Updated Nodes:]
For $i \in \{1,2\}$, define $\nodesize$-length vectors $\left\{\xis{\ell}{i}\right\}_{\ell \in [k]}$ as%\vspace{-.5cm}
\[
\begin{bmatrix}
\xis{1}{i}^T & \cdots & \xis{k}{i}^T \end{bmatrix}
:=\genmsrs{i}_s
\begin{bmatrix}
\genmsr_{\up_1}\\
\vdots\\
\genmsr_{\up_k}
\end{bmatrix}^{-1}
\]

Updated node $u_\ell$ ($\ell \in [k]$), which stores the updated data $\genmsr_{\up_\ell} \boldsymbol{m}'$, returns the two symbols:
\[
\xis{\ell}{1}^T \genmsr_{\up_\ell} \boldsymbol{m}' \qquad \textrm{and} \qquad \xis{\ell}{2}^T\genmsr_{\up_\ell} \boldsymbol{m}'
\]
%As discussed above, a set of $k$ nodes can transfer one symbol each to allow recovery of an arbitrary linear combination of the updated message. Thus, transferring two symbols each will allow for two arbitrary linear combinations.

%Let these linear combinations be $c_1, c_2$, corresponding to

%\begin{alignat*}{2}
%c_1^Tm' = \sum_{\p=1}^\np \etas{1}{\p} \boldsymbol\genmsr_{\st}^T M'_\p \tphi{\st}{1}
%\quad &&
%c_2^Tm' = \sum_{\p=1}^\np \etas{2}{\p} \boldsymbol\genmsr_{\st}^T M'_\p \tphi{\st}{2}
%\end{alignat*}
%
%Each updated node $\up_i$ contacted by the \stale node returns two symbols
%$h_{\up_i}(c_1), h_{\up_i}(c_2)$
%\vspace{.3cm}
%%\settowidth{\itemindent}{\bf Stale Node:}
\item[\bf Stale Node:]
Stale node $\st$, which stores stale data $\genmsr_s \boldsymbol{m}$,
performs the following operations.

\begin{enumerate}
\item
From the set of $2k$ received symbols, compute
$\sum_{\ell=1}^k \xis{\ell}{1}^T \genmsr_{\up_\ell} \boldsymbol{m}'
= \genmsrs1_s \boldsymbol{m}'$
and
$\sum_{\ell=1}^k \xis{\ell}{2}^T \genmsr_{\up_\ell} \boldsymbol{m}'
= \genmsrs2_s \boldsymbol{m}'$

%\item From the stale stored data, extract symbols $\genmsrs1_s \boldsymbol{m}$
%and $\genmsrs2_s \boldsymbol{m}$

\item Given the stale stored data, containing
$\genmsrs1_s \boldsymbol{m}$
and $\genmsrs2_s \boldsymbol{m}$,
take differences to obtain
$d_1 := \genmsrs1_s (\boldsymbol{m}' - \boldsymbol{m})$
and
$d_2 := \genmsrs2_s (\boldsymbol{m}' - \boldsymbol{m})$

If the changed symbol is at location $j$ in the message vector,
and its value has been changed by $\sdelta$, then
$d_1 = (\genmsrs1_s)_j \sdelta$
and
$d_2 = (\genmsrs2_s)_j \sdelta$

\item If $d_1 = d_2 = 0$ then the stale node already has the updated data; exit

\item Compute the ratio
$d_1 : d_2 =
 (\genmsrs1_s)_j :  (\genmsrs2_s)_j$

\item By construction, this ratio is unique for different values of $j$, so use the ratio to
identify the location $j_0$ of the change.

\item Compute $\sdelta =  ((\genmsrs1_s)_{j_0})^{-1} d_1$

\item Construct a $B$-length vector $\boldsymbol{\delta}$
with value $\sdelta$ at location $j_0$ and zeros elsewhere;
update the stale data by computing
$\genmsr_\st \boldsymbol{m}' = \genmsr_\st \boldsymbol{m} +
\genmsr_\st \boldsymbol{\delta}$

\end{enumerate}
\end{itemize}
\caption{Optimal Oblivious Update in an MDS Code}\label{algo:comm_MSR}
\end{algorithm}

\begin{theorem}
The code constructed in Section~\ref{sec:msr_encoding} is maximum-distance-separable (MDS).
\end{theorem}
\begin{proof}
Each node stores only $\frac{B}{k}$ symbols, and since every submatrix of $\genmsr$ is of full rank, the entire message is recoverable from any $k$ of the nodes.
\end{proof}

\section{Summary and Open Problems}
%\vspace{-.1cm}
%Traditional erasure coding literature for update complexity consider the centralized setting where the updated fragments corresponding to the modified data are computed at the source or a central node. In this paper, we considered a setting where the system is completely oblivious to the modified data and the stale nodes update their contents by downloading data from already updated nodes. This oblivious update setting is well suited for distributed systems such as peer-to-peer networks and data-center based storage systems. This paper has only started the investigation on fundamental limits on communication required for oblivious updates by establishing lower bounds and presenting code constructions meeting these lower bounds.

This paper considered the problem of \textit{oblivious updates} wherein the data stored in a storage node needs to be updated by downloading data from already updated nodes in the storage network, but with none of the nodes knowing the identity or the value of the modified data symbols. Oblivious updates allow the system to ensure that all nodes have the updated data (even after being offline/unavailable) without having to keep a log of modifications. We established the fundamental limits on the communication required for performing such oblivious updates, when a single message symbol is modified, by deriving genie-aided lower bounds and designing storage codes and update algorithms meeting these bounds. Our goal for the future is to extend the characterization of the fundamental limits in multiple directions, such as considering oblivious updates for multiple symbol modifications, non-linear codes, and interactive update protocols. %We also plan to explore codes that allow the stale node to detect when more than a threshold number of message symbols are modified, and perform optimal updates when fewer symbols are modified.
In addition, to complement the theoretical standpoint of this paper, we also plan to investigate the questions that arise in practical implementations of oblivious update protocols, such as the design of explicit codes, and quantification of the minimal state that needs to be maintained for realizing the update algorithms.
%\vspace{-.7cm}

%, and letting the \stale node to connect to more than $k$ nodes in the MDS setting.

\bibliographystyle{IEEEtran}
\bibliography{bibtex_updates}
\end{document}